\documentclass[conference]{IEEEtran}
\IEEEoverridecommandlockouts
\usepackage[T1]{fontenc}
\usepackage{dsfont}

\usepackage{graphicx}
\usepackage{placeins}
\usepackage{color,soul}
\usepackage{amsthm}
\usepackage{amsmath}
\usepackage{mathtools}
\usepackage{amsfonts}
\usepackage{cite}
\usepackage{amsmath,amssymb,amsfonts}
\usepackage{algorithmic}
\usepackage{graphicx}
\usepackage{textcomp}
\usepackage{xcolor}
\usepackage{color,soul}
\usepackage{amsthm}
\usepackage{subfigure}
\allowdisplaybreaks

\usepackage{mathtools}
\usepackage{verbatim}
\usepackage{bbold}

\def\dd{\mathrm{d}}

\DeclareMathOperator*{\argmin}{arg\,min}
\mathchardef\mhyphen="2D
\usepackage{newtxtext,newtxmath}
\usepackage{caption}
\usepackage{subfigure}
\usepackage{tikz}

\newtheorem{corollary}{Corollary}
\newtheorem{lemma}{Lemma}
\newtheorem{example}{Example}
\newtheorem{remark}{Remark}
\newtheorem{definition}{Definition}

\newcounter{MYtempeqncnt}

\def\ind{\mathds{1}}
\def\BibTeX{{\rm B\kern-.05em{\sc i\kern-.025em b}\kern-.08em
    T\kern-.1667em\lower.7ex\hbox{E}\kern-.125emX}}
\begin{document}
\title{ Performance Guarantees of Cellular Networks with Hardcore Regulation and Scheduling}

\author{\IEEEauthorblockN{Ke Feng}
\IEEEauthorblockA{\textit{CNRS-ETIS, ENSEA, CY Cergy Paris Université}, \\
ke.feng@cnrs.fr\\}
\and
\IEEEauthorblockN{Fran\c cois Baccelli}
\IEEEauthorblockA{\textit{INRIA-ENS and Telecom Paris},\\
francois.baccelli@inria.fr\\}
\and
\IEEEauthorblockN{Catherine Rosenberg}
\IEEEauthorblockA{\textit{University of Waterloo},\\
cath@uwaterloo.ca
}
}

\maketitle

\begin{abstract}
Providing performance guarantees is one of the {critical}  objectives of {recent and future} communication networks, toward which regulations, {i.e., constraints on key system parameters,} have played an indispensable role. This is the case for large wireless communication networks, where spatial regulations (e.g., constraints on intercell distance) have recently been shown, through a spatial network calculus, to be essential for establishing provable wireless link-level guarantees. In this work, we focus on performance guarantees for {the downlink of} cellular networks where we impose a hardcore (spatial) regulation on base station (BS) locations and evaluate {how BS scheduling (which controls which BSs can transmit at a given time) impacts performance}. Hardcore regulation is the simplest form of spatial regulation that enforces a minimal distance between any pair of transmitters in the network. Within this framework of spatial network calculus,  we first provide an upper bound on the power of total interference for a spatially regulated cellular network, and then, identify the regimes where scheduling BSs  yields {better} link-level rate guarantees compared to scenarios where base stations are always active. The hexagonal cellular network is analyzed as a special case. The results offer insights into what spatial regulations are needed, when to choose scheduling, and how to potentially reduce the network power consumption {to provide a certain target performance guarantee}. 
\end{abstract} 


\section{Introduction}
 { The main drive for strict delay guarantees in cellular networks comes from the development of wireless mission-critical communication systems \cite{Nasrallah19rll},
by analogy with what already exists for wireline networks, with, e.g., deterministic Ethernet \cite{bouillard2018deterministic}.
There are many other motivations for  guarantees in cellular networks, such as the deployment of broadband access via Fixed Wireless Access offering minimum bit rates to households \cite{andrew}.
A spatial network calculus was recently introduced in \cite{FengBaccelli2024spatial}, where it was shown that spatial regulation leads to a class of provable  performance guarantees \textit{for all links} in an arbitrarily large and even
infinite wireless network. In contrast, wireless networks lacking such a regulation cannot provide link-level performance guarantees.
This result opens up questions about what performance guarantees are feasible under different power control and coordinated transmission schemes \cite{zhong2025spatialnetworkcalculusdeterministic,sciancalepore2018multi,6953066}. This is particularly relevant  to cellular networks, where such schemes are often employed to improve the reliability and efficiency of communications.

The present work investigates spatial regulations for cellular networks in the form of \textit{hardcore regulation} and the performance guarantees {on the downlink} obtained by scheduling base stations (BSs). 
Hardcore regulation refers to a simple and effective way of implementing spatial regulation, that is obtained through imposing a distance constraint between active transmitters (in this case, BSs) and thus capping interference \cite{FengBaccelli2024spatial}, as in the setting of carrier sense multiple access with collision avoidance (CSMA/CA) for wireless networks \cite{Nguyen07-80211,zhong2024dualzonehardcoremodelrtscts}.
BS scheduling can be implemented by for example periodically muting certain BSs on predefined time-frequency resources, which is a methodology for inter-cell interference control (ICIC) and quality of service (QoS) control in dense cellular systems \cite{sciancalepore2018multi}. In terms of spatial regulation, BS scheduling allows one to decompose the entire set of BSs into several subsets of BSs, each potentially subject to stronger hardcore constraints and hence less severe interference. {Hence, the research question is: Can BS scheduling improve the rate guarantees on the downlink of cellular networks?}

In this work, we investigate the research question using spatial network calculus. We first provide a link-level upper bound  on interference {by taking  the user association (i.e., the process by which a user is mapped to a BS for ensuring its connectivity)} into account. The bound is shown tighter than what is previously available in \cite{FengBaccelli2024spatial}. 
Secondly, we show that the answer to {our question} depends on the transmit power, or alternatively, the signal-to-interference-plus-noise ratio (SINR), and the {BS deployment}. We further identify the regimes where scheduling improves the performance guarantees provided by spatial network calculus when BSs are always active {(AA)}. In such regimes, we show that one can gain by reducing the power consumption of the network. Lastly, we apply the analysis to hexagonal cellular networks.

The remainder of the paper is organized as follows. Section \ref{sec: back} introduces the concepts of spatial regulations that is key to this paper along with extensions and results. The system model is described in Section~\ref{sec: sys-model}.
Section~\ref{sec: performance} presents the main results for general networks while Section~\ref{sec:hex} focuses on hexagonal cellular networks. Section~\ref{sec:concl} summarizes the main findings and proposes avenues for future research.



\section{{Spatial Regulations: Background and Extensions} }\label{sec: back}

{In this section, we introduce   the  concepts of spatial regulations which are key to this paper along with some extensions and results that we introduce. For a comprehensive view of spatial network calculus, see \cite{FengBaccelli2024spatial}.}

Let $\Phi$ be a stationary point process on $\mathbb{R}^2$ defined on the probability space $(\Omega, \mathcal{A},\mathbb{P})$, {which models the locations of BSs in this work}. Let $b(x,r)$ denote the open ball of radius $r>0$, centered at $x\in \mathbb{R}^2$. Let $\Phi(b(x,r))\in\mathbb{N}\triangleq\{0,1,2,...\}$ denote the number of points of $\Phi$ residing in $b(x,r)$. $\|\cdot\|$ denotes Euclidean distance. 
We say that an event holds $\mathbb{P}\mhyphen\mathrm{a.s.}$ (which stands for almost surely) to indicate that the probability of this event is 1. 
\begin{definition}[Strong $(\sigma,\rho,\nu)$-Ball Regulation \cite{FengBaccelli2024spatial}]
\label{def: ball-reg}
    A stationary point process $\Phi$ is strongly $(\sigma,\rho,\nu)$-ball regulated if, for all $R\geq0$,
\begin{equation}
   \Phi(b(o,R))\leq \sigma +\rho R+\nu R^2,\quad \mathbb{P}\mhyphen a.s.\label{eq: def-ball-reg},\nonumber
\end{equation}
where $\sigma,\rho,\nu$ are  constants and $\sigma,\nu>0$.
\end{definition}

Note that this definition is some form of two dimensional extension of the
$(\sigma,\rho)$ regulation of classical network calculus \cite{cruz}. The bound in Definition \ref{def: ball-reg} also applies to $\Phi(b(y,R)), \forall y\in\mathbb{R}^2$ due to the stationarity of $\Phi$. This ball regulation is shown in \cite{FengBaccelli2024spatial} to be a sufficient and necessary condition for link-level performance guarantees.

A stationary point process $\Phi$ with hardcore distance $H$ is a stationary process in which any two distinct points are separated by a distance of at least
$2H$. In what follows, we will also refer to such point processes as being \textit{$H$-hardcore regulated}.

\begin{definition}[$(K,H_K)$-Hardcore Regulation]
\label{def: k,h_k}
Let $\Phi$ be a stationary point process. For $K\in\mathbb{N}^+$ and $H_K\in\mathbb{R}^+$, we say that $\Phi$ is $(K,H_K)$-hardcore regulated
if there exists a stationary marking $\{M(x)\}_{x\in\Phi}$ where $M(x)\in\{1,2,...,K\}$ 
such that for all $x,y\in\Phi$ with $M(x)=M(y)$,
\[\|x-y\|\geq 2H_K, \quad \mathbb{P}\mhyphen a.s.\]
\end{definition}
Here are a few remarks about  Definition \ref{def: k,h_k}:
\begin{itemize}
    \item {A stationary marking (see e.g. \cite{BBK}) can be a function of the local geometry. Stationarity means that if one shifts the whole point process, the mark of any given point is kept unchanged.} 

  \item   Definition \ref{def: k,h_k}  extends the notion of hardcore point process. A stationary $H$-hardcore regulated point process is a special case of a $(K, H_K)$-hardcore regulated process with $K=1$, i.e., it is $(1, H)$-hardcore regulated. {As we will show in Sections \ref{sec: sys-model} and \ref{sec: performance}, this new definition provides a useful property of point processes, which allows one to evaluate their performance with scheduling. }

\item If $\Phi$ is $(K,H_K)$-hardcore regulated, it is also $(K+1,H_{K})$-hardcore regulated. Alternatively speaking, for a given $\Phi$, $H_K$ monotonically increases with $K$.

\item Evidently, for a $(K,H_K)$-hardcore regulated point process and all markings satisfying the condition in Definition \ref{def: k,h_k}, $\forall i\in\{1,2,...,K\}$, the point process $\{x\in\Phi\colon M(x)=i\}$ is $H_K$-hardcore regulated.

    \item  The Poisson point process is not $(K,H_K)$ hardcore regulated for any $K\in\mathbb{N}^+$ and $H_K\in\mathbb{R}^+$.   However, it is possible to impose $(K,H_K)$-hardcore regulation on any point process through the process of thinning. {By thinning, we mean retaining only certain points, e.g., by discarding all those violating the hardcore condition.}
\end{itemize}

 Let $\rho_H \triangleq 2\pi/(\sqrt{12}H)$ and
$\nu_H \triangleq \pi/(\sqrt{12}H^2)$.
\begin{lemma}
\label{lemma: hardcore}
  Let $\Phi$ 
be $(K,H_K)$-hardcore regulated. Then there exists a stationary marking  such that $\forall i\in\{1,2,...,K\}$, $\{x\in\Phi\colon M(x)=i\}$ is strongly $(1,\rho_{H_K},\nu_{H_K})$-ball regulated; further, for all non-negative,
 bounded, and non-increasing functions $\ell\colon \mathbb{R}^{+}\to\mathbb{R}^{+}$, and for all $R>0$,
\begin{align*}
   &\sum_{x\in\Phi\cap b(o,R),M(x)=i} \ell(\|x\|)  \\
   &\leq \ell(0)+ \rho_{H_K} \int_{0}^{R}{\ell(r)}\dd r
    +2 \nu_{H_K} \int_{0}^R r\ell(r)\dd r,\quad \mathbb{P}\mhyphen a.s.\nonumber
\end{align*}.
\end{lemma}
\begin{proof} The first part of the statement follows from the definition of $(K,H_K)$-hardcore regulation and  \cite[Lemma 3]{FengBaccelli2024spatial}. The second part follows from \cite[Theorem 1]{FengBaccelli2024spatial}.
\end{proof}

Since $H_K$ monotonically increases with $K$ and the upper bound in Lemma \ref{lemma: hardcore} monotonically decreases in $H_K$, Lemma \ref{lemma: hardcore} implies that there is less interference in a scheduled (in this case, better separated) network.

\section{System Model}
\label{sec: sys-model}
\subsection{System Model}
Let $\Phi$ be a stationary point process on $\mathbb{R}^2$
modeling the locations of BSs in a cellular network.  Assume that the maximum transmit power of BSs is $P>0$, and that the isotropic signal attenuation is captured by a bounded, non-negative, and non-decreasing function $\ell\colon \mathbb{R}^+\to\mathbb{R}^+$. Let the additive white Gaussian noise (AWGN) power be denoted by $W>0$. We focus on the no-fading case and note that the analysis can be generalized to arbitrary fading using \cite[Theorem 2]{FengBaccelli2024spatial}.

We consider downlink transmission where the bandwidth is normalized and assume that users are associated to their closest BS, which is also known as the strongest-average-received-power association scheme. The cellular structure is therefore determined by the Voronoi tessellation of $\Phi$. For simplicity, let us further assume that each base station has one user, which can be generalized easily to the scenario where the number of users per cell is upper bounded by a fixed constant \cite[Definition 9]{FengBaccelli2024spatial}. Without loss of generality, {we consider a user at the origin} and denote its associated BS as $x_0\triangleq \argmin\{\|x\|\colon x\in\Phi\}$,
and let $d\triangleq\|x_0\|$. 

\subsection{Scheduling Schemes with Maximum Transmit Power}
Let the transmission slots be divided uniformly into periods of $K\in\mathbb{N}^+$ slots. We consider the following two cases: always active (AA) and periodic scheduling. In the former, all BSs are active in all slots. This is equivalent to the baseline regime studied in \cite{FengBaccelli2024spatial}. In the latter, in every period of $K$ slots, each base station is active and only active in one of slots determined by the scheduler. 

In both cases, define the normalized rate as
\begin{align*}
    {R} \triangleq \frac{1}{K} \sum_{i=1}^{K}{\log(1+{\rm{SINR}}_i)},
\end{align*}

where ${\rm{SINR}}_i$ the signal-to-interference-plus-noise ratio (SINR) of the user at the origin in the $i$-th slot.
The rate is normalized over the period of $K$ slots for a fair comparison of the two schemes.  

\subsubsection{Always Active}
For the AA case,
conditioned on $x_0$,  
\begin{equation}
{\rm{SINR}}_i = \frac{P\ell(d)}{P\sum_{x\in\Phi\setminus \{x_0\}}\ell(\|x\|)+W}.\nonumber
\end{equation}
${\rm{SINR}}_i$ does not change with $i$.
\subsubsection{Periodic Scheduling}
With periodic scheduling, the scheduler assigns to each BS $x\in\Phi$ a mark $M(x)$ from the index set $\{1,2,...,K\}$. $M(x)=i$ indicates that the BS $x$ is active in the $i$-th slot and muted in other slots. 
 
Conditioned on $x_0$, 
\begin{equation}
{\rm{SINR}}_i = \frac{P\ell(d)\ind(M(x_0)=i)}{P\sum_{x\in\Phi\setminus \{x_0\},M(x)=i}\ell(\|x\|)+W},\nonumber
\end{equation}
where ${\rm{SINR}}_i$ is the SINR of the user at the $i$-th slot and $\ind(\cdot)$ is the indicator function that is equal to 1 if the event happens and 0 otherwise.  If $x$ is not scheduled in the $i$-th slot, then ${\rm{SINR}}_i=0$. It is hence sufficient to consider the slot when the BS-user link is active.

\subsection{Scheduling with Reduced Power}
{
We extend the scheduling schemes described above, i.e., AA and periodic scheduling, to enable reduced BS transmit power. In this case, we denote the BS transmit power in the active slot by $P_K\leq P$. 
}
\section{Performance Guarantees of Cellular Networks}
\label{sec: performance}
In this section, we first derive a new upper bound on the interference for all users in hardcore regulated cellular networks. The bound is shown to be tighter than the existing bound in \cite{FengBaccelli2024spatial}. It is then used to derive lower bounds on the normalized rate for both the AA and periodic scheduling cases.

\subsection{Conditional Upper Bounds on Interference}
Let $\Phi$ be an $H$-hardcore regulated point process modeling BS locations. Consider the nearest BS-user association scheme, where given $x_0$, all other BSs are at least as far as $d=\|x_0\|$ from the origin. There is hence an exclusion region 
of interfering BSs, which we denote by $E $. We have
$$E = b(o,d) \cup b(x_0, 2H),$$
where the exclusion region $b(o,d)$ is due to the BS-user association scheme and the exclusion region $b(x_0, 2H)$ is due to the $H$-hardcore regulation of $\Phi$. This region is free from interferers, which leads to a tighter bound on interference. 

 Let $t \triangleq\max(d,{2{H}}-d)$. Then we have $b(o,t)\subset E$. Hence any interfering BS is at least at distance $t$.

\begin{figure}[t]
\centering
\includegraphics[width=0.95\linewidth]{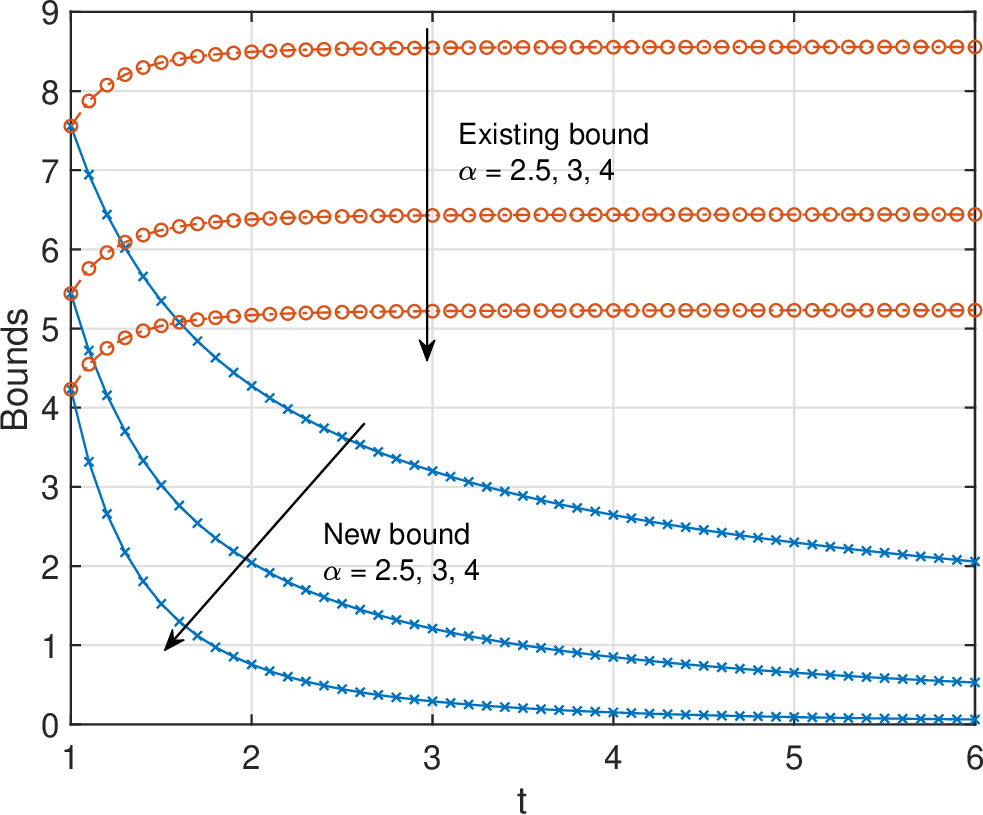}
    \caption{Comparison on the new upper bound on the total interference in Lemma \ref{lemma: conditional reg} versus the previous bound in \cite{FengBaccelli2024spatial} with path loss exponent $\alpha=2.5, 3, 4$, $P=1$, and $\ell(r)=\min\{1,r^{-\alpha}\}$.}
    \label{fig: inteference-bound}
\end{figure}

\begin{lemma} Let $\Phi$ be a $H$-hardcore-regulated  point process and let $\|x_0\|=d,$ Let $G:\mathbb{R}^+\to\mathbb{R}^+$ be a non-negative and non-decreasing function 
with countable discontinuities such that 
\[\Phi\left({b(o,R)\setminus E}\right) \leq G(R),\quad \mathbb{P}\mhyphen a.s.\]
	 For $R\geq t$, we have that, $\mathbb{P}\mhyphen a.s.$,
\begin{equation}
   \sum_{x\in\Phi\cap b(o,R)\setminus E}\ell(\|x\|) \leq -\int_{t}^{R} G(r)\ell'(r)\dd r +\ell(R) G{(R)}.\nonumber
\end{equation}
\label{lemma: conditional reg}
\end{lemma}

\begin{proof}
{{The point of $\Phi$ (excluding $x_0$) which is the closest from the receiver}}
is at distance ${t} =\max(d,{{2H}}-d)$. 
Let $N$ be a positive integer and $\Delta  = (R-t)/N $. Then $ \mathbb{P}\mhyphen a.s.,$
\begin{align}
&\sum_{x\in\Phi \cap b(o,R)\setminus E} \ell(\|x\|)\nonumber\\
&     \leq  
\ell(t)(\Phi(b(o,t+\Delta))-\Phi(b(o,t)))\nonumber\\
&\quad+
\ell(t+\Delta)(\Phi(b(o,t+2\Delta))-\Phi(b(o,t+\Delta)))+...\nonumber\\
&\quad+\ell(R-\Delta) (\Phi(b(o,R))-\Phi(b(o,R-\Delta)))\nonumber\\
& = -\ell(t)\Phi(b(o,t))+(\ell(t)-\ell(t+\Delta))\Phi(b(o,t+\Delta))+...\nonumber\\
&\quad+(\ell(R-2\Delta)-\ell(R-\Delta))\Phi(b(o,R-\Delta))\nonumber\\
&\quad+ \ell(R-\Delta)\Phi(b(o,R))\nonumber\nonumber\\
& \leq \sum_{n=0}^{N-1}(\ell(t+n\Delta)-\ell(t+(n+1)\Delta))\Phi(b(o,t+(n+1)\Delta)) \nonumber\\
&\quad+\ell(R)\Phi(b(o,R))\nonumber\\
& \leq \sum_{n=0}^{N-1}(\ell(t+n\Delta)-\ell(t+(n+1)\Delta))G(t+(n+1)\Delta) \nonumber\\
&\quad+\ell(R)G(R)\nonumber\\
& \xrightarrow{N\to\infty}  {\int_{t}^{R} -G(r)\ell'(r)\mathrm{d}r+\ell(R)G(R)}.\nonumber
\end{align}
\end{proof}

    To find the best (in the sense of achievability) $G(R)$, for every $R>0$, one needs 
the densest number of packed circles with radius $H$ in $b(o,R+H)$ with prohibited region $E$, 
a problem which is studied in discrete optimization \cite{pack,lopez2019packing}.

 For tractability, in later analyses, we will apply a slightly modified bound from Lemma \ref{lemma: hardcore}, namely $G(R) =\frac{2\pi}{\sqrt{12}H}R+\frac{\pi}{\sqrt{12}H^2}R^2$. In other words, $\sigma=0$. The previous constant $\sigma=1$ is removed to to account for the absence of the associated BS.

\begin{corollary}
\label{coro: bound-sn}
Let $\Phi$ be a stationary point process. If $\Phi$ is $H$-hardcore regulated, then for $R\geq t$,
\begin{align}
	&\sum_{x\in\Phi\cap b(o,R),x\neq x_0}\ell(\|x\|) \nonumber\\
    &\leq \int_{t}^{ R} \frac{2\pi\ell(r) }{\sqrt{12}}\left(\frac{1}{H}+\frac{r}{H^2}\right)\mathrm{d}r+\frac{\pi\ell(t)t}{\sqrt{12}}\left(\frac{2}{H}+\frac{t}{H^2}\right),\quad \mathbb{P}\mhyphen a.s.
\label{eq: new_bound_I}
\end{align}

    \end{corollary}
\begin{proof}
This comes from integration by part, Lemma \ref{lemma: conditional reg},  and Lemma \ref{lemma: hardcore}.
\end{proof}
\begin{example}

For $\ell(r)=\min\{1,r^{-\alpha}\}$, where $\alpha$ is the path loss exponent and  given $\|x_0\|=d$, from (\ref{eq: new_bound_I}) and letting $R\to\infty$, we have 
\begin{align*}
    &\sum_{x\in\Phi, x\neq x_0}\ell(\|x\|)\\
   & \leq\frac{2\pi}{\sqrt{12}}\frac{t^{1-\alpha}}{H}\frac{\alpha}{\alpha-1} + \frac{\pi}{\sqrt{12}}\frac{t^{2-\alpha}}{H^2}\frac{\alpha}{\alpha-2},\quad \mathbb{P}\mhyphen a.s.
\end{align*}

An earlier bound on the accumulated interference which does not take into account  the exclusion region $E$ is given in \cite[Corollary 3]{FengBaccelli2024spatial} as 
\begin{align*}
 &\sum_{x\in\Phi, x\neq x_0}\ell(\|x\|)\\
  &\leq \ell(0) + \rho_H \int_{0}^{\infty}\ell(r)\dd r + 2\nu_H\int_{0}^{\infty}\ell(r)r\dd r-\ell(t) \\
 &= 1+ \frac{2\pi}{\sqrt{12}H}\frac{\alpha}{\alpha-1}+ \frac{\pi}{\sqrt{12}H^2}\frac{\alpha}{\alpha-2} - \min\{1,t^{-\alpha}\}.
\end{align*}
\end{example}

Fig. 1 compares these two bounds for $H=1$, $d=1$, and $\ell(r)=\min\{1,r^{-\alpha}\}$ with $\alpha=2.5, 3, 4$. Observe that the new bound monotonically decreases with $t$, which indicates that there is less interference as the size of the exclusion region expands. In contrast, the earlier bound in \cite{FengBaccelli2024spatial} fails to take the impact of the exclusion region into account. {This improved bound is instrumental in the results of the next sections.}

\begin{figure}[t]
    \centering
\includegraphics[width=\linewidth]{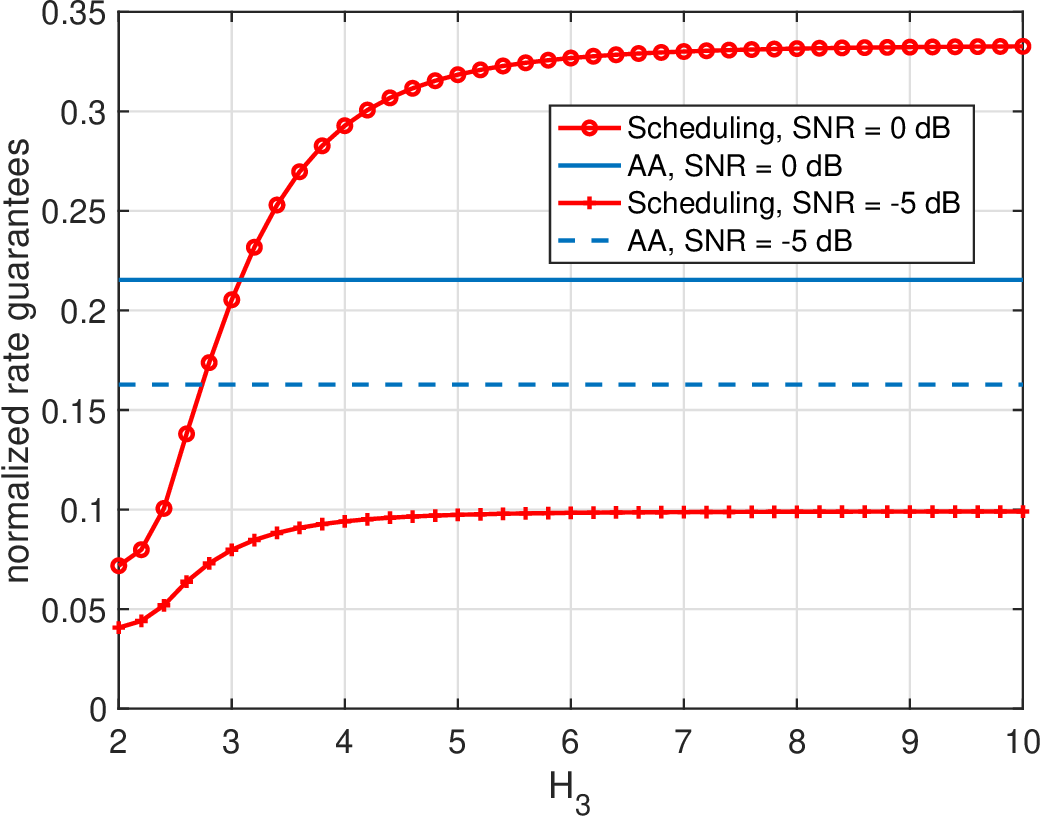}
    \caption{For $K=3$, guaranteed normalized rates for scheduling and always active (AA) vs. $H_3$. $H_1=2$. $d=4/\sqrt{3}$.  $\ell(r)=\min\{1,r^{-4}\}$.}
    \label{fig: k=2}
\end{figure}
\begin{figure}[t]
    \centering
\includegraphics[width=\linewidth]{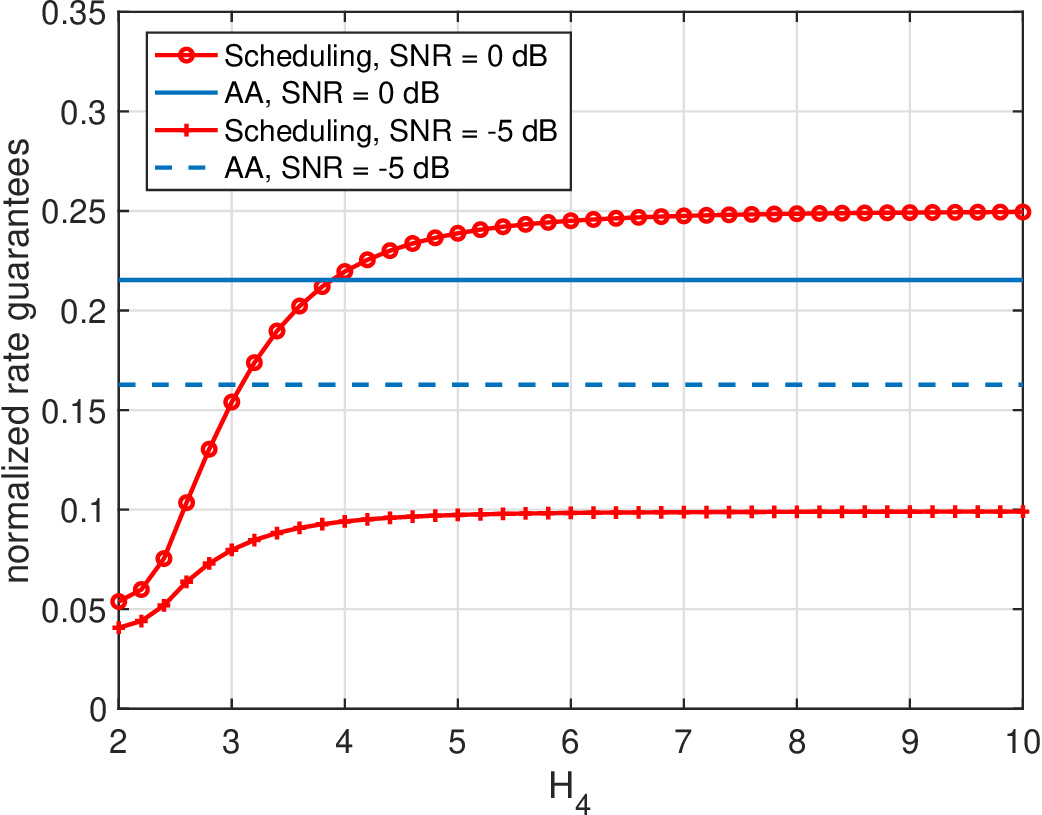}
    \caption{For $K=4$, guaranteed normalized rates for scheduling and always active (AA) vs. $H_4$. $H_1=2$. $d=4/\sqrt{3}$. $\ell(r)=\min\{1,r^{-4}\}$.}
    \label{fig: k=3}
\end{figure}

\subsection{Scheduling Schemes with Maximum Transmit Power}
Consider the cellular setting as described in Section \ref{sec: sys-model} and assume that all active BSs transmit with maximum transmit power $P$. Let $\Phi$ be $H$-hardcore regulated and $(K,H_K)$-hardcore regulated.  We assume that, for the given $K$, either $H_K$ or a lower bound of $H_K$ is known.

\subsubsection{Always Active (AA)}
\begin{corollary}
If $\Phi$ is $H$-hardcore regulated and BSs are always active, given $\|x_0\|=d,$  $\forall i\in\{1,2,...,K\}$, we have
\begin{equation}
\label{eq: sinr-bound}
    {\rm{SINR}}_i\geq  \theta(P,H),\quad \mathbb{P}\mhyphen a.s.\nonumber
\end{equation}
and $R\geq \log(1+\theta(P,H)), \mathbb{P}\mhyphen a.s.,$
where
\begin{align}
&\theta(P,H)\nonumber\\
&\triangleq
\frac{P\ell(d)}{P\int_{t}^{\infty} \frac{2\pi\ell(r) }{\sqrt{12}}\left(\frac{1}{H}+\frac{r}{H^2}\right)\mathrm{d}r+P\frac{\pi\ell(t)t}{\sqrt{12}}\left(\frac{2}{H}+\frac{t}{H^2}\right)+W}.
\label{eq: theta}
\end{align}
\end{corollary}
\begin{proof}
 Follows from the definition of the SINR and the upper bound 
 in Corollary \ref{coro: bound-sn} by letting $R\to\infty$.
\end{proof}
In this case, given $d$ and $P$, a sufficient $H$-hardcore regulation for $\theta$ to be a lower bound (in the a.s. sense) for the SINR is that $H$ is chosen such that $\theta(P,H)=\theta$. For such an $H$ to exist, one needs that $P\ell(d)/W>\theta$.

 \subsubsection{Periodic Scheduling}
 Let $\Phi$ be $(K,H_K)$-hardcore regulated. Then there exists a scheduling algorithm such that given $\|x_0\|=d,$ and $M(x_0)=i$, we have
\begin{equation}
\label{eq: sinr-bound}
    {\rm{SINR}}_i\geq  \theta(P,H_K),\quad \mathbb{P}\mhyphen a.s.\nonumber
\end{equation}
   \[{R} \geq  \frac{1}{K}\log\left(1+\theta(P,H_K)\right),\quad \mathbb{P}\mhyphen a.s.\]
  
By the monotonicity of $H_K$ in $K$, with scheduling, the SINR in the active slot is improved compared to the AA case. On the other hand, for  the normalized rate, there exists a tradeoff between the reduced number of allocated slots and the improved SINR due to better separation. This tradeoff is captured by the following.
    \begin{figure}[t]
    \centering
\includegraphics[width=.48\textwidth]{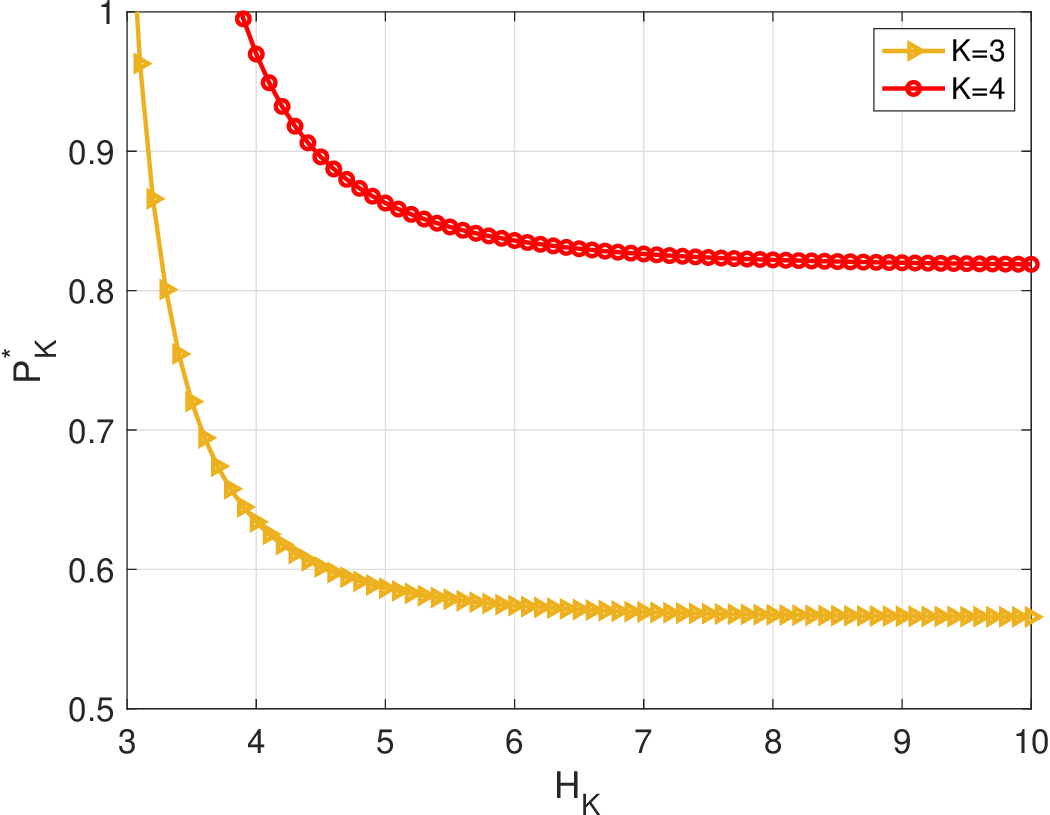}
    \caption{Reduced power $P_K^*$ vs. hardcore distance $H_K$ as in Eq (\ref{eq: P_critical}). $P=1$. SNR = 0dB. $H_1=2$. $d=4/\sqrt{3}$. $\ell(r)=\min\{1,r^{-4}\}$.}
    \label{fig:critical_P_K}
\end{figure}

\begin{lemma}
\label{def: criticality} 
Let $\Phi$ be $H$-hardcore regulated. For a given $K\in\mathbb{N}^+$,  there exists a unique solution $H_K^*\in\mathbb{R}^+$ to the following equation
\begin{equation}
\frac{1}{K}{\log(1+\theta(P,H_K))}=  {\log(1+\theta(P,H))},
\label{eq: criticality}
\end{equation}    
if
\begin{equation}
    \frac{\log\left(1+{\rm{SNR}}\right)}{\log(1+\theta(P,H))} \geq K,\nonumber
\end{equation}
where $\theta(P,H)$ is given in (\ref{eq: theta})
{and SNR is the signal-to-noise ratio  of the user given by
$${\rm{SNR}}\triangleq \frac{P\ell(d)}{W}.$$ }
\end{lemma}
\begin{proof}
     This is because $\log(1+\theta(P,H_K))\leq \log\left(1+{\rm{SNR}}\right) =\log(1+{P\ell(d)}/{W})$ for all $H_K$, and the equality is achieved when $H_K=\infty$.
    For the noise-free (or interference-limited) case, critical regulation always exists since $\rm{SNR} = \infty$ and $\log(1+\theta(P,H_K))$ is positive. The uniqueness of the solution, if it exists, follows from the monotonicity of $\theta(P, H_K)$ with respect to $H_K$.
\end{proof}

We refer to the solution of (\ref{eq: criticality}) in terms of $H_K$ as the critical hardcore regulation distance for $K$, denoted by $H_K^*$. 

\begin{figure*}[t!]
\normalsize
\setcounter{MYtempeqncnt}{\value{equation}}
\setcounter{equation}{\value{equation}}
    \begin{equation}
    \label{eq: P_critical}
        P_K^*=\frac{W}{\frac{\ell(d)}{(1+\theta(P,H))^K-1}-\int_{t}^{\infty} \frac{2\pi\ell(r) }{\sqrt{12}}\left(\frac{1}{H_K}+\frac{r}{H_K^2}\right)\mathrm{d}r-\frac{\pi\ell(t)t}{\sqrt{12}}\left(\frac{2}{H_K}+\frac{t}{H_K^2}\right)}.
    \end{equation}
    \end{figure*}
\begin{remark}
    Lemma \ref{def: criticality} identifies the regulation regime where scheduling can yield the same normalized rate guarantee as that of AA, hence the notion of criticality. For a given $\Phi$ that is $(K,H_K)$-hardcore regulated, if $H_K>H_K^*$, scheduling can provide a better guaranteed normalized rate; Otherwise, if $H_K<H_K^*$, then scheduling degrades the guaranteed normalized rate.
\end{remark}

Figs. \ref{fig: k=2} and \ref{fig: k=3} show the performance guarantees of the normalized rate for $H_1=2$, $d=4/\sqrt{3}$, in the relatively high SNR regime (SNR = 0dB), and relatively low SNR regime (SNR = -5dB), with $K=3,4$ respectively. $\ell(r)=\min\{1,r^{-\alpha}\}$ where $\alpha=4$. In the higher SNR regime, there exists an operation region of $H_K$ such that for $K=3,4$, scheduling under stronger hardcore regulation yields better rate guarantees. For example, the critical hardcore regulations for $K=3,4$ are $H_3^*
\approx3.15$ and $H_4^*\approx 3.8$. In contrast, in the lower SNR regime, AA always yields better bound regardless of $K,H_K$.
Intuitively, scheduling eliminates the interference from nearby BSs, which is particularly beneficial when interference is the primary limiting factor for the SINR. In contrast, when noise is more dominant, the normalization factor introduced by scheduling outweighs the benefits of reduced interference, hence reducing the normalized rate. 

\begin{figure}[t]
    \centering
\includegraphics[width=.48\textwidth]{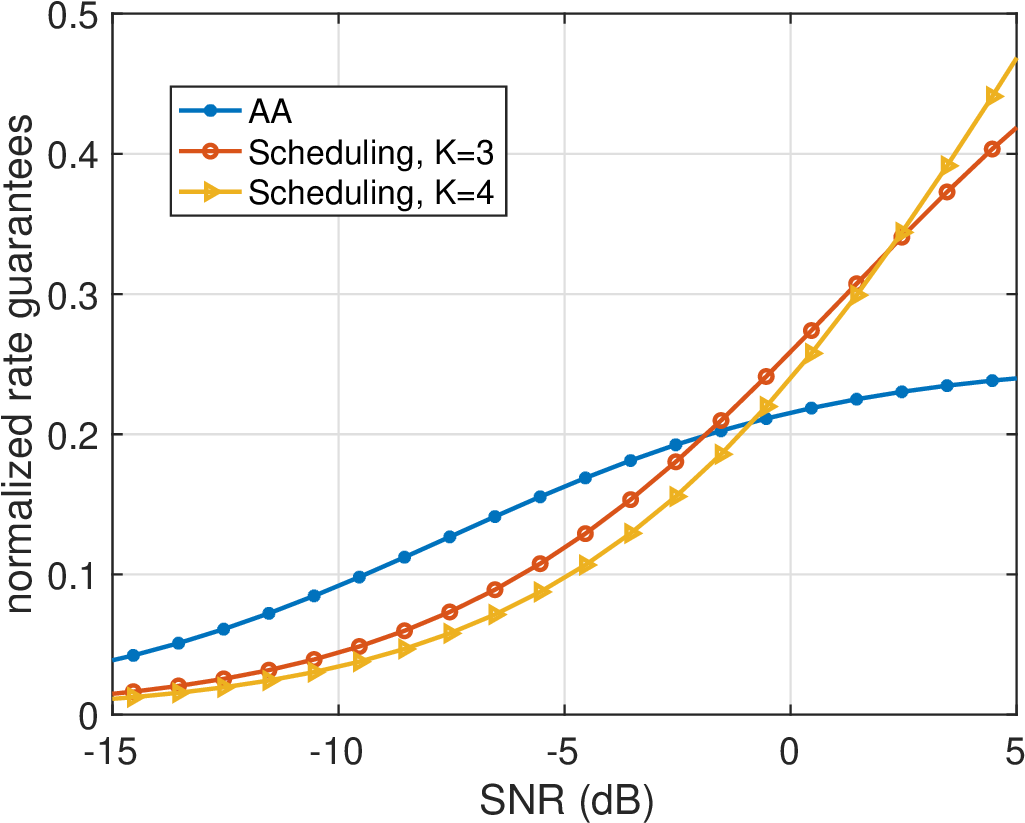}
    \caption{Normalized rate guarantees for scheduling
and always active (AA) for the hexagonal cellular structure vs. SNR. $a=4/\sqrt{3}, H_1=2$. $\ell(r)=\min\{1,r^{-4}\}$.}
    \label{fig:hex-bounds-1}
\end{figure}
\subsection{Scheduling Schemes with Reduced Power} 
\begin{corollary}
\label{cor: P_critical}    
For a given $H,P$ and $K$, if $\Phi$ is $(K,H_K)$-hardcore regulated and such that for $H_K\geq H_K^*$, the transmit power can be reduced provided $P_K\geq P_K^*$ without lowering the normalized rate guarantee provided by always active, where $P_K^*$ is given in Eq (\ref{eq: P_critical}).
\end{corollary} 
\begin{proof}
Eq (\ref{eq: P_critical}) follows directly from solving $\log(1+\theta(P_K^*,H_K)) = K\log(1+\theta(P,H))$. The condition $H_K\geq H_K^*$ guarantees that $P_K^*\leq P$. \end{proof}

\begin{remark}
   From Corollary \ref{cor: P_critical}, scheduling is a potential methodology to reduce the power consumption of a cellular network while maintaining desired rate guarantees.
\end{remark}
Fig. \ref{fig:critical_P_K} plots Eq (\ref{eq: P_critical}), the reduced power $P_K^*$  vs. hardcore distance for equal normalized rate guarantees with scheduling $K=3,4$ respectively, for $P=1$, SNR = 0dB, $H_1=2$, $d=4/\sqrt{3}$, and $\ell(r)=\min\{1,r^{-4}\}$.

\section{Hexagonal Cellular Networks}\label{sec:hex}
Let the locations of BSs $\Phi$ follow a triangular lattice, resulting in a hexagonal cellular structure. Let the edge length of the cells be $a$ and $\ell(r)=\min\{1,r^{-\alpha}\}$. For the lower bound on SINR for all users in the network, we focus on the worst-case user located at a cell vertex, i.e., $d=a$.
We compare the performance bounds under always active, periodic BS scheduling with three slots, and periodic BS scheduling with four slots with different SNR. For $K=1$, $H_1=\sqrt{3}a/2$. For $K=3$ and $K=4$, $H_3 = 3a/2$ and $H_4 = 2\sqrt{3}a$, the scheduling patterns follow the well-known optimal colorings of the hexagonal lattice using three and four colors, respectively.

 Fig. \ref{fig:hex-bounds-1} illustrates the normalized rate guarantees across different SNR levels at the vertex user. We take $P=1$, $a=4/\sqrt{3}$, $H_1=2$, and $\ell(r)=\min\{1,r^{-4}\}$. In the high-SNR regime, for example, at SNR $=0$ dB, $H_3=2\sqrt{3}>H_3^*\approx 3.15$ and $H_4=4>H_4^*\approx3.8$, hence scheduling with $K=3$ or $K=4$ improves the guaranteed normalized rate; in this regime, $P_3^*=0.7315$ and $P_4^*=0.9698$. When the SNR increases further, scheduling with $K = 4$ provides the best lower bound on the normalized rate.  In contrast, in the low-SNR regime, it is optimal for BSs to remain always active.

\section{Discussions and Conclusions}\label{sec:concl}

This work gives an improved upper bound on interference and lower bound on rate guarantees for all users in cellular networks with hardcore regulation and scheduling. Within the framework of spatial network calculus, it reveals the potential gain of scheduling in terms of rate guarantees and power reduction. 
There are interesting future directions for this work. The first is to study whether a critical $(K,H_K^*)$ hardcore regulation can be implemented for a given network deployment. If not, it would be valuable to determine the optimal regulation—via thinning—that maximizes BS intensity while satisfying either a target $(K, H_K)$ constraint or a desired performance lower bound. 
 Another direction is to incorporate traffic patterns, whether deterministic or stochastic, to study how scheduling interacts with traffic dynamics and influences overall network capacity. Lastly, since this work focuses on the no-fading scenario, it is essential to extend the performance analysis to account for more general fading scenarios.

\section*{Acknowledgements}
The work of F.B. and K.F. was supported by the European
Research Council project titled NEMO, under grant ERC
788851, and by the French National Agency for
Research project titled France 2030 PEPR réseaux du Futur
under grant ANR-22-PEFT-0010.

\bibliographystyle{IEEEtran}

\bibliography{ref,bookref}

@ARTICLE{andrew,
  author={Lappalainen, Andrew and Zhang, Yuhao and Rosenberg, Catherine},
  journal={IEEE Transactions on Network and Service Management}, 
  title={Planning {5G} Networks for Rural Fixed Wireless Access}, 
  year={2023},
  volume={20},
  number={1},
  pages={441-455},
  }

@book{BBK,
    author = {Baccelli, Francois and Blaszczyszyn, Bartek and Karray, Mohamed},
    title = {Random Measures, Point Processes, and Stochastic Geometry},
    publisher = {Online preprint, https://hal.inria.fr/hal-02460214},
    year = 2024
}

@INPROCEEDINGS{Nguyen07-80211,
  author={Nguyen, H. Q. and Baccelli, F. and Kofman, D.},
  booktitle={IEEE INFOCOM 2007 - 26th IEEE International Conference on Computer Communications}, 
  title={A Stochastic Geometry Analysis of Dense {IEEE} 802.11 Networks}, 
  year={2007},
  volume={},
  number={},
  pages={1199-1207},
  doi={10.1109/INFCOM.2007.143}}

@article{sciancalepore2018multi,
  title={A multi-traffic inter-cell interference coordination scheme in dense cellular networks},
  author={Sciancalepore, Vincenzo and Filippini, Ilario and Mancuso, Vincenzo and Capone, Antonio and Banchs, Albert},
  journal={IEEE/ACM Transactions on Networking},
  volume={26},
  number={5},
  pages={2361--2375},
  year={2018},
  publisher={IEEE}
}

@misc{zhong2024dualzonehardcoremodelrtscts,
      title={Dual-Zone Hard-Core Model for {RTS/CTS} Handshake Analysis in {WLANs}}, 
      author={Yi Zhong and Zhuoling Chen and Wenyi Zhang and Martin Haenggi},
      year={2024},
      eprint={2412.09953},
      archivePrefix={arXiv},
      primaryClass={cs.NI},
      url={https://arxiv.org/abs/2412.09953}, 
}

@misc{zhong2025spatialnetworkcalculusdeterministic,
      title={Spatial Network Calculus: Toward Deterministic Wireless Networking}, 
      author={Yi Zhong and Xiaohang Zhou and Ke Feng},
      year={2025},
      eprint={2501.02556},
      archivePrefix={arXiv},
      primaryClass={cs.NI},
      url={https://arxiv.org/abs/2501.02556}, 
}

@ARTICLE{Nasrallah19rll,  author={Nasrallah, Ahmed and Thyagaturu, Akhilesh S. and Alharbi, Ziyad and Wang, Cuixiang and Shao, Xing and Reisslein, Martin and ElBakoury, Hesham},  journal={IEEE Communications Surveys   Tutorials},   title={Ultra-Low Latency ({ULL}) Networks: The {IEEE TSN} and {IETF DetNet} Standards and Related {5G ULL} Research},   year={2019},  volume={21},  number={1},  pages={88-145},  doi={10.1109/COMST.2018.2869350}}

@book{bouillard2018deterministic,
  title={Deterministic network calculus: From theory to practical implementation},
  author={Bouillard, Anne and Boyer, Marc and Le Corronc, Euriell},
  year={2018},
  publisher={John Wiley \& Sons}
}

@article{pack,
author = {Y. Stoyan and G. Yaskov},
title = {Packing equal circles into a circle with circular prohibited areas},
journal = {International Journal of Computer Mathematics},
volume = {89},
number = {10},
pages = {1355--1369},
year = {2012},
publisher = {Taylor \& Francis},
doi = {10.1080/00207160.2012.685468},


URL = { 
    
        https://doi.org/10.1080/00207160.2012.685468
    
    

},
eprint = { 
    
        https://doi.org/10.1080/00207160.2012.685468
    
    

}

}

@ARTICLE{FengBaccelli2024spatial,
  author={Feng, Ke and Baccelli, François},
  journal={IEEE Transactions on Wireless Communications}, 
  title={Spatial Network Calculus and Performance Guarantees in Wireless Networks}, 
  year={2024},
  volume={23},
  number={5},
  pages={5033-5047},
  doi={10.1109/TWC.2023.3324069}}

@article{lopez2019packing,
  title={Packing a fixed number of identical circles in a circular container with circular prohibited areas},
  author={L{\'o}pez, Claudia O and Beasley, John E},
  journal={Optimization Letters},
  volume={13},
  pages={1449--1468},
  year={2019},
  publisher={Springer}
}

@ARTICLE{cruz,
  author={Cruz, R.L.},
  journal={IEEE Transactions on Information Theory}, 
  title={A calculus for network delay. {I}. Network elements in isolation}, 
  year={1991},
  volume={37},
  number={1},
  pages={114-131},
  doi={10.1109/18.61109}}

@ARTICLE{6953066,
  author={Lee, Namyoon and Lin, Xingqin and Andrews, Jeffrey G. and Heath, Robert W.},
  journal={IEEE Journal on Selected Areas in Communications}, 
  title={Power Control for {D2D} Underlaid Cellular Networks: Modeling, Algorithms, and Analysis}, 
  year={2015},
  volume={33},
  number={1},
  pages={1-13},
  doi={10.1109/JSAC.2014.2369612}}
\end{document}